\documentclass[conference,letterpaper]{IEEEtran}

\usepackage[utf8]{inputenc} 
\usepackage[T1]{fontenc}
\usepackage{url}
\usepackage{ifthen}
\usepackage{cite}
\usepackage[cmex10]{amsmath} 

\usepackage{enumerate}
\usepackage{amsthm}
\usepackage{latexsym}
\usepackage{amssymb}
\usepackage{xspace}
\usepackage{amscd}
\usepackage{graphicx}
\usepackage{epic}

\newcommand{\PAut}{\rm PAut}
\newcommand{\N}{{\mathbb N}}
\newcommand{\A}{\mathbb A}

\newcommand{\Z}{{\mathbb Z}}
\newcommand{\F}{\mathbb F}

\newcommand{\D}{{\mathcal D}}

\newcommand{\C}{{\mathcal C}}

\newcommand{\I}{{\mathcal I}}
\newcommand{\U}{{\mathcal U}}
\newcommand{\Db}{\overline{\D^*}}
\newcommand{\Or}{\mathcal{O}}

\newcommand{\doble}[2]{\genfrac{}{}{0cm}{2}{#1}{#2}}

\newtheorem{theorem}{Theorem}
\newtheorem{lemma}[theorem]{Lemma}
\newtheorem{proposition}[theorem]{Proposition}
\newtheorem{definition}[theorem]{Definition}
\newtheorem{corollary}[theorem]{Corollary}
\newtheorem{remark}[theorem]{Remark}
\newtheorem{remarks}[theorem]{Remarks}

\newtheorem{example}[theorem]{Example}

\begin{document}
\title{Generalized Reed-Muller codes: A new construction of information sets}

\author{%
  \IEEEauthorblockN{José Joaquín Bernal}
  \IEEEauthorblockA{Departmento de Matemáticas\\
                    Universidad de Murcia\\
                    30100 Murcia, SPAIN\\
                    Email: josejoaquin.bernal@um.es}
}

\maketitle

\begin{abstract}
   In \cite{BS} we show how to construct information sets for Reed-Muller codes only in terms of their basic parameters. In this work we deal with the corresponding problem for $q$-ary Generalized Reed-Muller codes of first and second order. We see that for first-order codes the result for binary Reed-Muller codes is also valid, while for second-order codes, with $q> 2$, we have to manage more complex defining sets and we show that we get different information sets. We also present some examples and associated open problems.
\end{abstract}

\section{Introduction}
In 1954, D. E. Muller introduced the family of (binary) Reed-Muller codes \cite{Muller}. In the same year, I. S. Reed presented a specific decoding algorithm for them \cite{Reed}. Later, Kasami, Lin and Peterson generalized this family of codes to other ground fields \cite{KLP}. Independently. E. J. Weldon presented the nonprimitive generalization and the single-variable approach using the Mattson-Solomon polynomial \cite{Weldon}. The reader may find an extensive explanation on both generalizations in \cite{DGM, handbook3} and  \cite{ handbook2}. In 1954

In this paper we consider primitive GRM codes seen as affine-invariant codes in the algebra $\F G$, where $\F$ is the field with $q$ elements and $G$ is the abelian group of the field with $q^m$ elements, for some natural number $m$. Then, as any affine-invariant code, GRM codes have an associated defining set (see Definition \ref{definingset} below). 

The goal of this paper is to give information sets for first and second order GRM codes only in terms of their defining sets and ultimately, in terms of their basic parameters, namely,  their order and their length. In \cite{BS} we reach such a description in the case of (binary) Reed-Muller codes; however, for GRM codes we deal with a more complex structure of their defining sets, specially in the case of second-order GRM codes.

In Sections \ref{preliminaries} and \ref{2dimcyclic} we give a basic introduction to abelian codes and we show the essential tools about cyclic codes seen as two-dimensional cyclic codes, respectively. Then, in Section \ref{2dimcyclic} we introduce briefly GRM codes as affine-invariant codes. In Section \ref{mainsection} we present our main results, with two subsections corresponding to first and second-order GRM codes respectively. Finally, in Section \ref{examples} we include some examples and 
in Section \ref{conclusions} we give our conclusions and we mention some related open problems.

\section{Preliminaries}\label{preliminaries}
We need to introduce some basic notation and definitions about abelian codes that will be needed all throughout this paper. For our purposes we only need to consider two-dimensional abelian codes so we are giving this introduction just in that case.
First, we will always denote by $\F$ the field of $q$ elements where $q$ is a power of a prime $p$. Then, we denote  
$$\A(r_1,r_2):=\F[X_1,X_2]/\left\langle X_1^{r_1}-1,X_2^{r_2}-1\right\rangle$$ 
and we consider that any (two-dimensional) abelian code is an ideal in $\A(r_1,r_2)$. We identify the codewords with polynomials $P(X_1,X_2)$ such that every monomial satisfies that the degree of the indeterminate $X_i$ is in $\Z_{r_i}$. We write $n=r_1\cdot r_2$ and we always assume that $r_i$ and $q$ are coprime, for $i=1,2$, that is, we assume that $\A(r_1,r_2)$ is a semisimple algebra.

We will base our construction on the study of the so-called defining set of an abelian code, so let us introduce that definition.

\begin{definition}\label{definingset}
Let $\C\leq \A(r_1,r_2)$ be an abelian code. Fixed $\alpha_i$ a primitive $r_i$-th root of unity in some extension of $\F$,  $i=1,2$, the \textit{defining} set of $\C$ with respect to $\alpha=\{\alpha_1,\alpha_2\}$ is 
$$ \begin{array}{rl}
D_\alpha\left(\C\right)=&\left\{ (a_1,a_2)\in  \Z_{r_1}\times\Z_{r_2} \mid\right.\\
&  \left.P(\alpha_1^{a_1},\alpha_l^{a_2})=0 \text{ for all } P\in \C\right\}.
\end{array}$$
\end{definition}

\begin{remark}\label{definingsetcyclic}
The simply adaptation of the previous definition to the case $\A(n)=\F[X]/\langle X^n-1\rangle$ coincides with the classical notion of defining set of a cyclic code of length $n$.
\end{remark}

It is clear that the defining set of a code depends on the election of the primitive roots of unity and it can be proven that fixed $\alpha$ any abelian code is completely determined by its defining set.

In particular, we are working with the description of defining sets as union of some especial subsets of integers. Let us introduce them.

\begin{definition}
 Let $a$ and $r$ be integers. The $q$-cyclotomic coset of $a$ modulo $r$ is the set 
\[ C_{q,r}(a)=\left\{a\cdot q^i \mid i \in \N\right\} \subseteq \Z_r.\]
\end{definition}

The following definition gives the generalization of $q$-cyclotomic coset to the two-dimensional space $\A(r_1,r_2)$.

\begin{definition}\label{qorbita}
Given $(a_1,a_2)\in \Z_{r_1}\times\Z_{r_2}$, its \textit{$q$-orbit} modulo  $\left(r_1,r_2 \right)$ is the set
$$ Q_q(a_1,a_2)=\left\{\left(a_1\cdot q^{i} , a_2\cdot q^{i}  \right)\mid i\in \N\right\} \subseteq \Z_{r_1}\times\Z_{r_2}.$$
\end{definition}

It is easy to see that, under our assumption $\gcd(r_1,q)=\gcd(r_2,q)=1$, the $q$-orbits define a partition in $\Z_{r_1}\times\Z_{r_2}$; moreover, any defining set is a (disjoint) union of $q$-orbits and conversely, any union of $q$-orbits define a defining set (with respect to a fixed $\alpha$) and so an abelian code. Since we will only consider $q$-cyclotomic cosets and $q$-orbits, for the sake of simplicity, we shall write $C_r(a)$ and $Q(a_1,a_2)$ without mention to $q$.

To do the required computations on the elements in $\D(\C)$ we will select a subset satisfying the following definition.

	\begin{definition}\label{restricted representatives}
Let $D$ be a union of $q$-orbits modulo $(r_1,r_2)$ and $\overline{D}\subset D$ a complete set of representatives. Then $\overline D$ is called a set of restricted representatives if $\overline{D}_1$  is a complete set of representatives of the $q$-cyclotomic cosets modulo $r_1$ in $D_1$ ($\overline{D_1}$ and $D_1$ the projections of $\overline{D}$ and $D$ onto $\Z_{r_1}$ respectively).
 \end{definition}

To finish this section we need to give the formal notion of information set in our ambient space $\A(r_1, r_2)$.

\begin{definition}
 An information set for an abelian code $\C\subseteq\A(r_1, r_2)$ with dimension $k$ is a set $\I\subseteq \Z_{r_1}\times\Z_{r_2}$ such that $\left|\I\right|=k$ and $\C_\I=\F^k$, where $\C_\I$ denotes the projection of the codewords in $\C$ to the positions corresponding with $\I$. 
 
 The complementary set $\left(\Z_{r_1}\times\Z_{r_2}\right)\setminus \I$ is called a set of check positions for $\C$.
\end{definition}

\section{Information sets for cyclic codes seen as two-dimensional abelian codes}\label{2dimcyclic}

According to the notation that has been introduced in the previous sections we take a cyclic code, denoted by $\C^*$, with length $n=r_1\cdot r_2$, and assume that $r_1,r_2>1, \gcd(r_1,r_2)=1$, $\gcd(r_i,q)=1$ for $i=1,2$. Then let $T:\Z_n\rightarrow \Z_{r_1}\times \Z_{r_2}$ be an isomorphism, and let us denote $T=(T_1,T_2)$; that is, $T_i(e)$ is the projection of $T(e)$ onto $\Z_{r_i}$, for $i=1,2$ and any $e\in \Z_n$. 

According to Remark~\ref{definingsetcyclic}, let $\D^*=D_\alpha(\C^*)\subseteq \Z_n$ be the defining set of $\C^*$ with respect to an arbitrary primitive $n$-th  root of unity $\alpha$. Then, since there exist integers $\eta_1,\eta_2$ such that $\eta_1 r_1+\eta_2 r_2=1$, we have that $\alpha_1=\alpha^{\eta_2 r_2}$ and $\alpha_2=\alpha^{\eta_1 r_1}$ are primitive $r_1$-th and $r_2$-th roots of unity respectively. Fix $T$ an isomorphism as above and set $T(1)=(\delta_1,\delta_2)$; observe that $\gcd(\delta_1,r_1)=1$ and $\gcd(\delta_2,r_2)=1$. We define the abelian code $\C=\C_{(\C^*,T)}\leq \A(r_1,r_2)$ as the code with defining set $\D=D(\C)=T(\D^*)$, with respect to $(\beta_1,\beta_2)=(\alpha_1^{\delta_1^{-1}},\alpha_2^{\delta_2^{-1}})$. In this situation, we have that $\C$ is the image of $\C^*$ by the map 
  \begin{eqnarray*}
  \nonumber \A(n) &\longrightarrow& \A(r_1,r_2)\\
   \sum a_{i}X^i&\hookrightarrow& \sum b_{jl}X^jY^l,
  \end{eqnarray*}
  where $b_{jl}=a_i$ if and only if $T(i)=(j,l)$. Therefore, $\I$ is an information set for $\C$ if and only if $T^{-1}(\I)$ is an information set for $\C^*$. Usually, we omit the reference to the original cyclic code and the isomorphism $T$ in the notation of the new abelian code, and we will write $\C$ instead of $\C_{(\C^*,T)}$; those references will by clear by the context.
  
  For the rest of this section we assume that we have fixed a choice of $\alpha$ and $T$ (consequently, the roots $\beta_1,\beta_2$ are also fixed).

\begin{definition}\label{suitable}
 Let $\overline {\D^*} \subseteq \D^*$ be a complete set of representatives of the $q$-cyclotomic cosets modulo $n$ in $\D^*$. Then $\overline{\D^*}$ is said to be a suitable set of representatives if $T(\overline{\D^*})$ is a set of restricted representatives of the $q$-orbits in $T(\D^*)$ (see Definition \ref{restricted representatives}).
\end{definition}

It can be proved that we will always be able to take a suitable set of representatives in $\D^*$ (see \cite{BS}).

The goal of this section is to give an information set for the abelian code $\C\leq \A(r_1,r_2)$ just in terms of the defining set of the original cyclic code, this is established in Theorem \ref{infosetabelian}. To get it we need to introduce some special subsets in $\overline{\D^*}$ on which we will do the computations. 

\begin{definition}\label{defU}
Given $\overline {\D^*}\subseteq \D^*$ a suitable set of representatives, we define $\sim$ the equivalence relation on $\overline {\D^*}$ given by the rule 
\begin{equation}\label{relacionU}
 a\sim b\in \overline {\D^*} \text{ if and only if } a\equiv b \mod r_1.
\end{equation}
  We denote by $\U \subseteq \Db$ a complete set of representatives of the equivalence classes given by $\sim$. Then, for any $u\in \U$ we write 
\begin{equation*}
\Or(u)=\{a\in\Db \mid a\sim u \}.
\end{equation*}
\end{definition}

\begin{remark}
It is easy to see that for any $e\in \overline{\D^*}$ there exists a unique $u\in \U$ such that $T_1(u)=T_1(e)$. By abuse of notation, we will write 
$$C_{r_1}(e)=C_{r_1}(T_1(e))=C_{r_1}(T_1(u))=C_{r_1}(u).$$
\end{remark}

Next result gives what we were aiming for in this section.

\begin{theorem}\label{infosetabelian}
In the setting described above, let $\C^*\leq \A(n)$ be a cyclic code of length $n$, $\Db$ a suitable set of representatives of its defining set, with respect to a given $n$-th root of unity, and $\U$ as in Definition \ref{defU}. Let $\C\leq \A(r_1,r_2)$ be the abelian code $\C_{(\C^*,T)}$.  In this situation we consider the values
 $$\left\{M(u)=\frac{1}{|C_{r_1}(u)|}\sum_{v\in \Or(u)}|C_n(v)| \mid u\in \U\right\}.$$
Then we define 
$$f_1=\max\limits_{u\in \U}\{M(u)\}, \qquad f_i=\max\limits_{u\in \U}\{M(u)\mid M(u)<f_{i-1}\}.$$
which yields the sequence $f_{1}>\dots>f_{s}>0=f_{s+1}$. On the other hand, for any $k=1,\dots,s$ we define
$$g_k=\sum_{\doble{u\in \U}{M(u)\geq f_k}}|C_{r_1}(u)|.$$
Then, the set 
 \begin{eqnarray*}\label{checkpositions}
 \Gamma(\C)=&\{(i_1,i_2)\in \Z_{r_1}\times\Z_{r_2}\mid  \text{ there exists } 1\leq j\leq s \text{ with }\\
& f_{j+1}\leq i_2< f_{j}, \text{ and } 0\leq i_1<g_{j}\}.
\end{eqnarray*}
is a set of check positions for $\C$. Furthermore, the set $T^{-1}\left(\Gamma(\C)\right)$ is a set of check positions for $\C^*$.
\end{theorem}

\begin{proof}
The reader may check that the proof is constructed by using some results in \cite{BS} and \cite{BS2}. Furthermore, it can be seen that the proofs do not depend on the values of $q$.

Specifically, we begin by taking Theorem 20 in \cite{BS2} which gives us the description of an information set for an arbitrary abelian code from some calculations on the cardinalities of the $q$-orbits in its defining set. Then we use Lemma 18 and Proposition 19 in \cite{BS2} which relate the cardinalities of  the cyclotomic cosets in the defining set of a cyclic code and the cardinalities of the $q$-orbits for the corresponding abelian code obtained as in the beginning of the previous section. Finally, Theorem 20 in \cite{BS} complete the computations we need.
\end{proof}

\section{Generalized Reed-Muller codes}\label{GRMcodes}

We are seeing GRM codes as affine-invariant codes (and, in particular,  as extended cyclic codes) (see \cite{handbook2}). Let us introduce this family of codes briefly. The ambient space will be the group algebra $\F G$, where $G$ is the additive group of the field with $q^m$ elements, for some $m\in \N$. All throughout we fixed $m$ and we set $n=q^m-1$ (this notation is being consistent with that given in the previous sections). 

Since $G$ is an elementary abelian group of order $|G|=q^m$ and $G^*=G\setminus\{0\}$ is a cyclic group, we may take a generator element $\langle \alpha \rangle=G^*$ (note that $\alpha$ is also a primitive $n$-th root of unity in the field extension of $\F$ with $q^m$ elements). Then, we write the elements in $\F G$ as
	\begin{equation}\label{polynomial}
	 b X^0+\sum\limits_{i=0}^{n-1} a_i X^{\alpha^i}\quad (a_i,b\in \F).
	\end{equation}
	
	On the other hand, we consider $S_G$, the group of automorphisms of $G$, acting on $\F G$ via 
		$$\tau\left(b X^0+\sum\limits_{i=0}^{n-1} a_i X^{\alpha^i}\right)=b X^{\tau(0)}+\sum\limits_{i=0}^{n-1} a_i X^{\tau(\alpha^i)}$$
	for any $\tau \in S_G$. Then, we define the permutation automorphisms group of a code (ideal) $\C\leq \F G$ as
	$$\PAut(\C)=\{\tau\in S_{G}\mid \tau(\C)=\C\}.$$

Now, we can give the definition of affine-invariant code.

\begin{definition}
Let $\C\leq \F G$ be a code:
\begin{enumerate}[a)]
	\item $\C$ is an extended cyclic code if for any $b X^0+\sum\limits_{i=0}^{n-1} a_i X^{\alpha^i}\in \C$ one has that $b X^0+\sum\limits_{i=0}^{n-1} a_i X^{\alpha^{(i+1)}}\in \C$ and $b +\sum\limits_{i=0}^{n-1} a_i =0\in \F$.
	\item $\C$ is an \emph{affine-invariant} code if it is an extended cyclic code and ${\rm GA}(G,G)\subseteq \PAut(\C)$ where
	$${\rm GA}(G,G)=\{x\mapsto ax+b\mid x\in\F G, a\in G^*, b\in G\}.$$
\end{enumerate} 
\end{definition}

For any affine-invariant code $\C\leq \F G$ we denote by $\C^*$ the punctured code at the position $X^0$. Then, $\C^*$ is a \textit{cyclic} code in the sense that	it is the projection to $\F G^*$ of the image of a cyclic code via the map 
 \begin{eqnarray}\label{inmersionciclicos}
  \nonumber \A(n) &\longrightarrow& \F G\\
   \sum\limits_{i=0}^{n-1}a_{i}X^i&\hookrightarrow& \left(-\sum\limits_{i=0}^{n-1} a_i\right)X^0+\sum\limits_{i=0}^{n-1}a_i X^{\alpha^i}.	  
  \end{eqnarray}
	
\bigskip	
To give our new construction of information sets we will work with the so-called defining set of an affine-invariant code. Let us note that this notion is not the same as it was given in Definition \ref{definingset}; we will clarify the difference below.

\begin{definition}  Let $\C\leq\F G$ be an affine-invariant code and $\langle\alpha\rangle=G^*$. For any $s\in\{0,\dots,n=q^m-1\}$ we consider the $\F$-linear map $\phi_s:\F G\rightarrow G$ given by
\begin{equation*}
 \phi_s\left(b X^0+\sum\limits_{i=0}^{n-1} a_i X^{\alpha^i}\right)=0^s+\sum\limits_{i=0}^{n-1} a_i \alpha^{is}
\end{equation*}
where we assume $0^0=1\in\F$ by convention. Then the set 
 $$D(\C)=\{i\mid \phi_i(x)=0 \text{ for all } x\in \C\}$$
 is called the defining set with respect to $\alpha$ of the affine-invariant code $\C$.
\end{definition}

It is clear that $D(\C)$ depends on the election of $\alpha\in G$, however all throughout this paper it will be fixed and an arbitrary one so, for the sake of simplicity, we do not include it in the notation. 

It is also well known that any affine-invariant code is totally determined by its defining set (see \cite{handbook2}). On the other hand, it is easy to see that $D(\C)$ is a union of $q$-cyclotomic cosets modulo $n$; moreover, since $\C$ is an extended-cyclic code, one has that $0\in D(\C)$. Conversely, any subset of $\{0,\dots,n\}$ which is a union of $q$-cyclotomic cosets and contains $0$ defines an affine-invariant code in $\F G$.

It essential for us to note that with regard to (\ref{inmersionciclicos}), we talk about the defining set of $\C^*$ when we are making reference to that of the corresponding cyclic code in $\A(n)$. So,  the defining set of $\C^*$ is $D(\C^*)=D(\C)\setminus\{0\}$ (according to Definition \ref{definingset}).

\begin{remark}\label{check positions union 0}
In the context of affine-invariant codes, a set $\I\subseteq\{0,\alpha^0,\dots,\alpha^{n-1}\}$  is an information set for a code $\C\leq\F G$ with dimension $k$, if $\left|\I \right|=k$ and $\C_\I=\F^{\left|\I \right|}$.
 
 It is also important to note that, since $\C^*$ is contained in $\F G^*$ we have that any information set for it will be a subset of $\{\alpha^0,\dots,\alpha^{n-1}\}$. Furthermore, an information set for $\C^*$ is an information set for $\C$ too. However, note that if $\Gamma$ is a set of check positions for $\C^*$ then a set of check positions for $\C$ is $\Gamma\cup\{0\}$.
\end{remark}

Finally, to give the definition of GRM codes we need to introduce the notion of $q$-weight of a natural number.

\begin{definition}
For any natural number $k$ its $q$-ary expansion is the sum 
$$\sum_{r\geq 0} k_r q^r=k$$
 with $k_r\in \{0,1,\dots,q-1\}$. The $q$-weight of $k$ is ${\rm wt}_q(k)=\sum_{r\geq 0}k_r$. 
\end{definition}

\begin{definition}\label{defRM}
 Let $0< \rho\leq m(q-1)$ and $\alpha$ a generator of $G^*$. The Generalized Reed-Muller (GRM) code of order $\rho$ and length $q^m$, denoted by $R_q(\rho,m)$, is the affine-invariant code in $\F G$ with defining set (with respect to $\alpha$) 
 $$D(R_q(\rho,m))=\{0\leq i<q^m-1 \mid {\rm wt}_q(i)< m(q-1)-\rho\}.$$
\end{definition}

In this paper we are interested in GRM codes of first and second-order respectively, for them
$$D(R_q(1,m))=\{0\leq i<q^m-1 \mid {\rm wt}_q(i)< m(q-1)-1\}.$$
$$D(R_q(2,m))=\{0\leq i<q^m-1 \mid {\rm wt}_q(i)< m(q-1)-2\}.$$

\begin{remarks}
\begin{enumerate}
 \item We denote by $R^*(\rho,m)$ the punctured code of $R_q(\rho,m)$ at the position $X^0$.
 \item It is well-known that the dual codes of $R_q(1,m)$ and $R_q(2,m)$ are $R_q(m(q-1)-2,m)$ and $R_q(m(q-1)-3,m)$ respectively (see \cite{Huffman}). This fact will be important for us in the subsequent section.
\end{enumerate}
\end{remarks}
 
\section{Information sets for Generalized Reed-Muller codes}\label{mainsection}

In this main section, we are applying the results of Section \ref{2dimcyclic} to the (cyclic) punctured code of first and second-order GRM codes seen as two dimensional abelian codes. Recall that we are assuming that $n=q^m-1=r_1\cdot r_2,\ r_1,r_2>1, \gcd(r_1,r_2)=1$. (Note that the condition $\gcd(r_i,q), i=1,2$ holds by definition of $n$). Let use denote $a=Ord_{r_1}(q)$, that is, the order of $q$ modulo $r_1$; then, it is easy to see that $a\mid m$.

Since the given construction of information sets relies on the structure of the defining sets, it is more convenient to look at the dual codes instead of the original ones because they have smaller defining sets and so they are more manageable; that is, for $i=1,2$, instead of considering the code $R_q(i,m)$ we work with its dual $R_q(m(q-1)-i-1),m)$. Then, Theorem \ref{infosetabelian} will give us a set of check positions for $R_q(m(q-1)-i-1),m)$ and consequently an information set for $R_q(i,m)$.

Let us denote $\C^*=R^*_q(m(q-1)-i-1,m), i=1,2,$ and define 
\begin{eqnarray*}
 \Omega(K)&=&\left\{0< j<  q^m-1\mid {\rm wt}_q(j)=K\right\}
\end{eqnarray*}

then 
$$D^*=D(\C^*)=\bigcup_{i\in D(R(\rho,m))\setminus\{0\}}C_n(i)$$
where $\rho=m(q-1)-i-1$.

We need to describe completely the structure of $D^*$ and firstly we have to study how we can take a suitable set of representatives and the corresponding set $\U$ (see Definition \ref{defU}).

As in Section \ref{2dimcyclic} we assume that we have fixed a choice of $\alpha$, a primitive $n$-th root of unity, and $T$ an isomorphism from $\Z_n$ to $\Z_{r_1}\times\Z_{r_2}$.

\subsection{First-order GRM codes}

In this case, 
$$D(R_q(m(q-1)-2,m))=\{0\leq i<q^m-1 \mid {\rm wt}_q(i)< 2\}$$
and so $D^*=\Omega(1)$.

\begin{lemma}
In the case of $R_q(m(q-1)-2,m)$ we can take $\Db=\U=\{1\}$.
\end{lemma} 
\begin{proof}
By definition 
$$\Omega(1)=\left\{0< j<  q^m-1\mid {\rm wt}(j)=1\right\}=\{q^i\mid 0\leq i<m\}$$
So, $\Omega(1)=C_n(1)$. Then $D^*$ has a unique $q$-cyclotomic coset and the result is straightforward.
\end{proof}

Therefore, we have the same situation as for (binary) Reed-Muller codes. Namely, from Theorem \ref{infosetabelian} we get our first main result.

\begin{theorem}\label{teorema1order}
  The set $T^{-1}\left(\Gamma\right)$, where 
  $$\Gamma=\left\{(i_1,i_2)\in\Z_{r_1}\times\Z_{r_2} \mid 0\leq i_1< a, 0\leq i_2<\frac{m}{a}\right\},$$
is a set of check positions for $R^*_q(m(q-1)-2,m)$. 
\end{theorem}

The reader may check that the value of $q$ has no influence in the proof of this theorem, so that given in \cite[Theorem 34]{BS} also works in this context.

\begin{corollary}\label{corolario1order}
Given $\Gamma$, the set of check positions described in the previous theorem, the set $\{0,\alpha^i\mid i\in T^{-1}\left(\Gamma\right)\}\subseteq G$ is an information set for $R_q(1,m)$ and $\{\alpha^i\mid i\notin T^{-1}\left(\Gamma\right)\}$ is an information set for $R_q(m(q-1)-2,m)$.
\end{corollary}
\begin{proof}
It follows from the last statement of Theorem \ref{infosetabelian}.
\end{proof}

The reader may see that we have obtained the same information set that was given in \cite{BS}, that is, it is valid for any value of $q\geq 2$.

\subsection{Second-order GRM codes}

In this section we will assume $q>2$ as it is a necessary condition for the results to come until the end of this paper. Therefore, the main results in this case are not valid for (binary) Reed-Muller codes. 

For second-order GRM codes we have
$$D(R_q(m(q-1)-3,m))=\{0\leq i<q^m-1 \mid {\rm wt}_q(i)< 3\}$$
and then
$$D^*=\Omega(1)\cup \Omega(2),$$
more specifically, if we write $B_1=\{2q^{i}\mid 0\leq i<m\}$ and $B_2=\left\{q^{i}+q^{j}\mid 0\leq i<j<m\right\}$ , we have $\Omega(2)=B_1\cup B_2$ and so
\begin{equation}\label{D-estrella}
  D^*=\Omega(1)\cup B_1\cup B_2.
\end{equation}
(Note that $B_1=C_n(2)$.)

Note that in this case we have that the structure of $D^*$ is more sophisticated. So, in order to be able to do the required computations we will impose the additional condition $r_1=q^a-1$ for some natural number $a$. (We will see in Section \ref{examples} that this additional condition is not a strong restriction.) 
Then the restrictions needed are
\begin{eqnarray}\label{restrictions}
  n=q^m-1=r_1\cdot r_2,\ r_1, r_2>1,\ \gcd(r_1,r_2)=1\\
	r_1=q^a-1, \text{ for some } a\in \N,\nonumber 	
\end{eqnarray}

\begin{remark}
By definition we have that the order of $q$ modulo $r_1$ is precisely $a$. Note that the notation is consistent with that given at the beginning of this section. Let us write 
$$m=ab.$$
\end{remark}

First, we need the following lemma.

\begin{lemma}\label{cardinales}
If $q>2$ ($q$ a power of a prime number $p$) and $r=q^a-1$ for some natural number $a$ then
\begin{enumerate}
	\item $C_r(1)\neq C_r(2).$
	\item $|C_r(1)|=|C_r(2)|=a.$
\end{enumerate}
\end{lemma}
\begin{proof}
To prove statement $1)$ suppose that $2\in C_r(1)$. Then, there exists $\delta\in\N$ such that $q^\delta\equiv 2$ mod $r$, that is, there exists $k\in\Z$ with $q^\delta-2=k\cdot (q^a-1)$. But, since $q$ is a power of a primer number greater than $2$, we have that $q^a-1$ is even which implies $q^\delta$ even, a contradiction. So, $2\notin C_r(1)$ and since the cyclotomic cosets define a partition we get what we wanted.

To prove statement $2)$ note first that $|C_r(1)|=a$ from the definition of $q$-cyclotomic coset modulo $r$. Now, let $\mu=|C_r(2)|$, that is, $\mu$ is the minimum positive integer such that $2\equiv q^\mu\cdot 2$ mod $r$. Let us denote $I=\{i\in \Z\mid 2\equiv q^i\cdot 2 \text{ mod } r\}$. Then, it is clear that $0,a\in I$. Moreover, given $i,j\in I$ we have that $2\equiv 2q^i$ mod $r$ implies (since $q$ is a unit in $\Z_r$) that $2q^{-i}\equiv 2$ mod $r$, $2q^{i+j}\equiv 2q^iq^j\equiv 2q^j\equiv 2$ mod $r$, and for any $h\in \Z, 2 q^{ih}\equiv 2q^iq^{i(h-1)}\equiv 2q^{i(h-1)}\equiv\dots\equiv 2q^i\equiv 2$ mod $r$. Therefore, $I$ is an ideal in $\Z$ so it is a principal ideal. Then, by definition of $\mu$ it is easy to see that $I$ is the ideal generated by $\mu$. This implies that $\mu\mid a$. Let $s\in \N$ be such that $a=\mu s$, then there exists $k\in \Z$ such that $2(q^\mu-1)=k\cdot(q^a-1)=$ so 
$$2(q^\mu-1)=k\cdot (q^\mu-1)(q^{\mu(s-1)}+q^{\mu(s-2)}+\cdots+q^\mu+1)$$
and hence
$$2=k \cdot (q^{\mu(s-1)}+q^{\mu(s-2)}+\cdots+q^\mu+1)$$
which is possible only in case $s=1$, that is, $\mu=a$. This finishes the proof.
\end{proof}

Now, we can prove that we will always may assume that the essential values $1$ and $2$ belong to the sets we have to manage in our computations.

\begin{proposition}
  Let $D^*$ be the defining set described in (\ref{D-estrella}). Then
		\begin{enumerate}
	  \item We may take $\overline{D^*}\subseteq D^*$, a suitable set of representatives, such that $1,2\in\overline{D^*}$. Moreover,	
		\item we may take $\U\subseteq \D^*$,  with $1,2\in\U$, and	
		\item $\Or(2)=\{2\}$.
	\end{enumerate}	
\end{proposition}
\begin{proof}
Note first that by Lemma \ref{cardinales} we have that $C_n(1)\neq C_n(2)$. Now, we can take $1,2$ in a suitable set of representatives, $\Db$, if and only if $T(1)$ and $T(2)$ belong to a restricted set of representatives  of the $q$-orbits in $T(D^*)$ if and only if $T(1)$ and $T(2)$ do not belong to the same $q$-cyclotomic coset modulo $r_1$; since $r_1=q^a-1$, this last condition holds again by Lemma \ref{cardinales}. Furthermore, this also implies that $1\not\equiv 2$ mod $r_1$ and then we can take them as elements in $\U$.

Finally, as we have seen, $1\notin \Or(2)$ and if $e\in \U\setminus\{1,2\}$ then $e\in B_2$, which implies that ${\rm wt}_q(e)=2$ and so $e\notin C_{r_1}(2)$. This means that $e\notin\Or(2)$. We can conclude that $\Or(2)=\{2\}$.
\end{proof}

From the previous result we have that we may treat separately the value $2\in \U$. This is a crucial fact in order to be able to use the computations made in \cite{BS}, this will be describe explicitly within the proof of Theorem \ref{main2order}.

To give our main result we only need one more lemma.

\begin{lemma}\label{cardinaldeU}
For any election of the set of representatives $\U$ one has that $$\left|\U\right|=2+\left|\{C_{r_1}(e)\mid e\in\U\cap B_2\}\right|=2+\left\lfloor a/2\right\rfloor.$$
\end{lemma}
\begin{proof}
From the previous results we have that $1,2\in\U$ and $\Or(2)=\{2\}$. On the other hand, note that 
$$\Db\setminus\{1,2\}\subseteq B_2.$$
Then, by \cite[Lemma 45]{BS}, we have that $\left|\left\{e\in\Db\setminus\{1,2\}\mid e\notin\Or(1)\right\}\right|=\left\lfloor a/2\right\rfloor$. This finishes the proof.
\end{proof}

\begin{theorem}\label{main2order}
Let $\C^*=R^*_q(m(q-1)-3,m)$ and suppose that the conditions given in (\ref{restrictions}) hold. Let $\D^*\subseteq \Z_n$ be the defining set of $\C^*$, with respect to $\alpha$, a primitive $n$-th root of unity, and let $\C\leq \A(r_1,r_2)$ be the abelian code with defining set $D(\C)=T(\D^*)$. We define the sets $\gamma_i\subseteq\Z_{r_1}\times\Z_{r_2}$, for $i=1,2,3$, given by

$\begin{array}{rl}
\gamma_1=&\displaystyle\left\{(i_1,i_2)\mid 0\leq i_1<\frac{a(a-1)}{2}\ \text{ and }\ 0\leq i_2<b^2\right\},\\
&\\
\gamma_2=&\left\{(i_1,i_2)\mid \displaystyle\frac{a(a-1)}{2}\leq i_1<\frac{a(a+1)}{2}\ \text{ and }\right.\\
   & \displaystyle \left.0\leq i_2<\frac{b(b+1)}{2}\right\},\\
	&\\
\gamma_3=&\displaystyle\left\{(i_1,i_2)\mid \frac{a(a+1)}{2}\leq i_1<\frac{a(a+3)}{2}\ \text{ and }\right.\\
 &\displaystyle \left.0\leq i_2<b\right\}.\\
&
\end{array}$

Then, the set $T^{-1}\left(\Gamma\right)$ where 
$$\Gamma=\gamma_1\cup\gamma_2\cup\gamma_3$$
 is a set of check positions for $R^*_q(m(q-1)-3,m)$. 
\end{theorem}

\begin{proof}
Let $\Db$ be a suitable set of representatives of the $q$-cyclotomic cosets modulo $n$ in $\D^*$ and let $\U$ be as in Definition \ref{defU}, such that $1,2\in\U$. We want to apply Theorem \ref{infosetabelian} to the code $\C$. To do it we need to compute the value $M(u)$ for any $u\in \U$. 

In \cite{BS} we obtain those values for any $u\in\U\setminus\{2\}$, to wit
$$M(1)=\frac{b(b+1)}{2}$$
and 
$$M(e)=b^2, \text{ for any } e\in\U\setminus\{1,2\}.$$
The reader may see that these computations are also valid for $q\neq 2$.

Now, let us compute $M(2)$. By definition
$$M(2)=\frac{1}{|C_{r_1}(2)|}\sum\limits_{v\in\Or(2)}|C_n(v)|=\frac{|C_n(2)|}{|C_{r_1}(2)|}=\frac{m}{a}=b$$
where we have used that $|C_n(2)|=|C_n(1)|$ and $|C_{r_1}(2)|=|C_{r_1}(1)|$ (see Lemma \ref{cardinales}). So, we have the sequence
$$f_1=b^2>f_2=\frac{b(b+1)}{2}>f_3=b$$
Observe that we have the restricted inequalities because $b>1$ (as $r_2>1$). From this, following the notation in Theorem \ref{infosetabelian} we have the other sequence
$$g_1=\sum_{\doble{u\in\U}{M(u)\geq f_1}}|C_{r_1}(u)|=\sum\limits_{u\in\U\setminus\{1,2\}}|C_{r_1}(u)|=\frac{a(a-1)}{2}$$
$$g_2=\sum_{\doble{u\in\U}{M(u)\geq f_2}}|C_{r_1}(u)|=\frac{a(a-1)}{2}+|C_{r_1}(1)|=\frac{a(a+1)}{2}$$
$$g_3=\sum_{\doble{u\in\U}{M(u)\geq f_3}}|C_{r_1}(u)|=\frac{a(a+1)}{2}+|C_{r_1}(2)|=\frac{a(a+3)}{2}$$
with
$$g_1<g_2<g_3.$$

Finally, a direct application of Theorem \ref{infosetabelian} gives us what we wanted.
\end{proof}

As in the previous section for first-order GRM codes we conclude with the following corollary.

\begin{corollary}\label{corolario2order}
Given $\Gamma$, the set of check positions described in the previous theorem, the set $\{0,\alpha^i\mid i\in T^{-1}\left(\Gamma\right)\}\subseteq G$ is an information set for $R_q(2,m)$ and $\{\alpha^i\mid i\notin T^{-1}\left(\Gamma\right)\}$ is an information set for $R_q(m(q-1)-3,m)$.
\end{corollary}

\section{Examples}\label{examples}

\subsection{Examples for first-order GRM codes}
First, we include in Table \ref{Tabla1order} examples of suitable decompositions of $n=q^m-1$, with $3\leq m\leq 10$, and $q=3,5$.

\begin{table}[ht]
\caption{Suitable decompositions for $n=q^m-1$ in first-order GRM codes}
\label{Tabla1order}
\centering

\begin{tabular}{|c|c|c|c||c|c|c|c|}
 \hline &&&&&&&\\
 $q$ & $m$ & $r_1$ & $r_2$ & $q$ & $m$ & $r_1$ & $r_2$\\
&&&&&&&\\ \hline
3&3&13&2& 5&3&31&4\\ \hline
..&4&5&16& ..&4&13&48\\ \hline
..&5&121&2& ..&5&11&284\\ \hline
..&6&7&104& ..&6&7&2232\\ \hline
..&7&1093&2& ..&7&19531&4\\ \hline
..&8&41&160& ..&8&313&1284\\ \hline
..&9&757&26& ..&9&19&102796\\ \hline
..&10&61&968& ..&10&33&295928\\ \hline
\end{tabular}
\end{table}

In the following example we show how to apply the construction of information sets for first-order GRM. 

\begin{example}
 Let $q=3, m=3$ so we look for an information set for $R_3(1,3)$. As we have seen in Table \ref{Tabla1order} a suitable decomposition of $n=3^3-1=26$ is $r_1=13, r_2=2$. Moreover, it is straightforward to see that $a=Ord_{13}(3)=3$. 

Now, we fix $\alpha$ a primitive $26$-th root of unity and the isomorphism $T:\Z_{26}\rightarrow \Z_{13}\times\Z_{2}$ given by $T(1)=(1,1)$. Then, by Theorem \ref{teorema1order} the set
$T^{-1}(\Gamma)$ where
$$\Gamma=\{(i_1,i_2)\in\Z_{13}\times\Z_{2}\mid 0\leq i_1<3, i_2=0\}$$
is a set of check positions for $R^*_3(4,3)$. It is easy to see that $T^{-1}(\Gamma)=\{0,2,14\}$. Moreover, by Corollary \ref{corolario1order} the set $\{0,1,\alpha^2,\alpha^{14}\}$ is an information set for $R(1,3)$. 

\end{example}

\begin{remark}\label{obs1order}
The reader may check that a different election of $\alpha$ is equivalent to a different election of $T$. Furthermore, different elections of $T$ may yield different information sets. So, with this construction we get at most as many information sets as there are isomorphisms.
\end{remark}

\subsection{Second-order GRM codes}

First, let us show that conditions (\ref{restrictions}) are not too restrictive. As examples, we take $q=3, 5$ and in Table \ref{Tabla2order} we show that we can find suitable values easily for $m\leq 10$. Recall that we need the condition $r_1=q^a-1$, for some natural number $a$, so we have included the column with the values $a$ instead the corresponding to $r_2$; the reader could obtain $r_2$ as $n/r_1$.

\begin{table}[ht]
\caption{Suitable decompositions for $n=q^m-1$ in second-order GRM codes}
\label{Tabla2order}
\centering

\begin{tabular}{|c|c|c|c||c|c|c|c|}
 \hline &&&&&&&\\
 $q$ & $m$ & $r_1$ & $a$ & $q$ & $m$ & $r_1$ & $a$\\
&&&&&&&\\ \hline
3&3&2&1& 5&3&4&1\\ \hline
..&5&2&1& ..&5&4&1\\ \hline
..&6&8&2& ..&7&4&1\\ \hline
..&7&2&1& ..&9&4&1\\ \hline
..&9&2&1& ..&9&124&3\\ \hline
..&9&26&3& ..&10&24&2\\ \hline
..&10&8&2& &&&\\ \hline
\end{tabular}
\end{table}

Now, let us show how the construction works. 

\begin{example}
 Let $q=5, m=3$ so we look for an information set for $R_5(2,3)$. As we have seen in Table \ref{Tabla2order} a suitable decomposition of $n=5^3-1=124$ is $r_1=4, r_2=31$, where $a=1$. Now, we fix $\alpha$ a primitive $124$-th root of unity and the isomorphism $T:\Z_{124}\rightarrow \Z_{4}\times\Z_{31}$ given by $T(1)=(1,1)$.

Then, we have the following sets, which were mentioned in Theorem \ref{main2order}:
$$\gamma_1=\emptyset$$
$$\gamma_2=\{(i_1,i_2)\in\Z_{4}\times\Z_{31}\mid i_1=0, 0\leq i_2<6\}$$
$$\gamma_3=\{(i_1,i_2)\in\Z_{4}\times\Z_{31}\mid i_1=1, 0\leq i_2<3\}$$

Then, by Theorem \ref{main2order} the set
$T^{-1}(\Gamma)$ where
$$\Gamma=\gamma_2\cup\gamma_3$$
is a set of check positions for $R^*_5(9,3)$. Moreover, by Corollary \ref{corolario2order} the set $\{0,\alpha^i\mid i\in T^{-1}\left(\Gamma\right)\}$ is an information set for $R(2,3)$.
\end{example}

Let us note that the same comment that we made in Remark \ref{obs1order} is valid in this case.

\section{Conclusions and related open problems}\label{conclusions}
In this work we have dealt with the problem of constructing information sets for GRM codes of first and second order just in term of their basic parameters. We have obtained that while for first-order GRM codes we get the same information sets that was presented in \cite{BS} for (binary) Reed-Muller codes, for second-order GRM codes we get different ones with a more complex structure.

In relation with this results some open problems arise:
\begin{itemize}
	\item It seems to be possible to study the approach to the aim of this paper seeing the punctured cyclic codes as multidimensional cyclic codes, that is, in the algebra $\A(r_1,\dots,r_l)$ with $l>2$. Meanwhile, it is clear that this drive us to a considerably increase of the complexity of the computations on the cyclotomic cosets.
	\item In \cite{BS3} we show how to apply the permutation decoding algorithm to first-order GRM codes and we obtain very relevant improvements in the number of errors to be corrected. So, it becomes a very interesting problem how we can apply that decoding algorithm to second-order GRM codes with respect to the new information sets described in the present work.
\end{itemize}

\section*{Acknowledgments}
This work has been supported by the Spanish Government under Grant PID2020-113206GBI00
funded by MCIN/AEI/10.13039/501100011033 and Fundación Séneca of Murcia, Project PI/22.

\end{document}